\documentclass[pdflatex,sn-mathphys-num]{sn-jnl}

\usepackage{graphicx}%
\usepackage{multirow}%
\usepackage{amsmath,amssymb,amsfonts}%
\usepackage{amsthm}%
\usepackage{mathrsfs}%
\usepackage[title]{appendix}%
\usepackage{xcolor}%
\usepackage{textcomp}%
\usepackage{manyfoot}%
\usepackage{booktabs}%
\usepackage{algorithm}%
\usepackage{algorithmicx}%
\usepackage{algpseudocode}%
\usepackage{listings}%
\usepackage{enumitem}
\usepackage{mathabx}
\usepackage{comment} 

\renewcommand{\Vec}[1]{\mbox{\boldmath $#1$}}
\renewcommand{\proofname}{\textbf{Proof.}}
\newcommand{\Prec}{\mbox{\rm Prc}}
\newcommand{\Merg}{\mathrm{Merg}}
\newcommand{\ComG}{{\rm K}}

\newcommand{\Exp}{\mathbb E}
\newcommand{\Prob}{\mbox{\rm Pr}}

\algnewcommand\Input{\item[\textbf{Input:}]} 
\algnewcommand\Output{\item[\textbf{Output:}]}

\newtheorem{Theorem}{Theorem}[section]

\newtheorem{Lemma}[Theorem]{Lemma}
\newtheorem{Cor}[Theorem]{Corollary}

\newtheorem{Assumption}[Theorem]{Assumption}

\theoremstyle{thmstyleone}%
\theoremstyle{thmstyletwo}%
\newtheorem{remark}[Theorem]{Remark}%

\theoremstyle{thmstylethree}%

\begin{document}

\title[Article Titl]{
Approximating the Shapley Value \\
of Minimum Cost Spanning Tree Games: \\
An FPRAS for Saving Games}


\author[1]{\fnm{Takumi} \sur{Jimbo}}

\author*[1]{\fnm{Tomomi} \sur{Matsui}}\email{matsui.t.af@m.titech.ac.jp}

\affil[1]{\orgdiv{Department of Industrial Engineering and Economics}, \\
\orgname{Institute of Science Tokyo}, 
\orgaddress{\street{Meguro-ku}, \city{Tokyo}, 
\postcode{152-8552}, 
\country{Japan}}}


\abstract{
In this research, we address the problem 
    of computing the Shapley value 
    in minimum‑cost spanning tree (MCST) games. 
We introduce the saving game as a key framework 
    for approximating the Shapley value. 
By reformulating MCST games 
    into their saving‑game counterparts, 
    we obtain structural properties 
    that enable multiplicative (relative‑error) approximation. 
Building on this reformulation, 
    we develop a Monte Carlo based 
    Fully Polynomial-time Randomized 
    Approximation Scheme (FPRAS) 
    for the Shapley value.
}



\maketitle

\vspace{-5ex}

\begin{center}
\large
\date{\today}
\end{center}

\medskip

\section{Introduction} \label{section:intro}

Minimum spanning tree problems arise in network design settings
where a group of agents, located at different nodes, require a service
provided by a common source and must share the cost of the connecting network.
Beyond the problem of constructing a minimum  spanning tree, a central
question concerns how the total connection cost should be allocated among
the participating agents.

Claus and Kleitman~\cite{claus1973cost} first examined cost allocation 
    in the minimum spanning tree.
Bird~\cite{bird1976cost} introduced a game-theoretic formulation 
    of minimum spanning tree problems by associating 
    a cooperative cost allocation game with each instance.
This line of work laid the foundation 
    for analyzing minimum spanning tree problems 
    by formulating them as 
    a minimum‑cost spanning tree game 
    (MCST game), 
    thereby enabling the use 
    of cooperative game‑theoretic solution concepts 
    to investigate fairness and stability 
    in cost sharing.   
Further theoretical developments 
    were provided in~\cite{granot1981minimum,granot1982relationship,granot1984core}.
Early works and reviews are due to~\cite{aarts1993irreducible,curiel1997minimum,feltkamp2000bird,borm2001operations}.
For more recent developments, 
    see the reviews~\cite{trudeau2013characterizations,trudeau2019shapley,bergantinos2021review}.

The Shapley value proposed in~\cite{shapley1953value}
    is a solution concept in cooperative game theory.
Despite its strong normative appeal, 
    computing the Shapley value is
    computationally challenging.
In particular, Ando~\cite{ando2012computation} showed that 
    the exact calculation of the Shapley
    value in MCST games is \#P-hard, 
    even under highly restricted settings,
    indicating that an efficient exact algorithm 
    is unlikely to exist in general.

To circumvent this computational difficulty, 
    existing studies have mainly
    followed two distinct directions.
The first line of research focuses 
   on identifying restricted classes of MCST
   games for which the Shapley value 
   can be computed exactly 
    in polynomial time~\cite{ando2010reduction,ando2012computation}.
%
The second line of research addresses approximation algorithms
   for general MCST games.
Ando and Takase~\cite{ando2020monte} proposed 
    Monte Carlo based sampling algorithms
    that approximate the Shapley value
    with additive-error (absolute-error) guarantees.
Furthermore, Takase and Ando~\cite{takase2019approximaiton} 
    introduced a deterministic polynomial-time
    approximation method based on structural approximations using chordal graphs.
While these approaches are applicable to general MCST games, 
    the former relies on random sampling 
    and provides only additive-error guarantees. 

In this study, we investigate Monte Carlo–based approximation 
    of the Shapley value in MCST games, 
    with a focus on achieving relative‑error guarantees. 
Existing Monte Carlo methods provide 
    only additive‑error bounds, 
    which can be inadequate when Shapley values 
    vary widely across players or instances. 
To address this issue, we introduce the saving game,  
    which is a framework that converts the cost formulation 
    into a value formulation 
    based on coalition savings. 
For the saving game associated with the MCST game, 
    this reformulation enables multiplicative (relative‑error) 
    approximation of the Shapley value. 
Computational experiments further show that 
    the proposed method exhibits empirical performance 
    consistent with an FPRAS.

The rest of this paper is organized as follows. 
Section~\ref{section:pre} introduces 
 the necessary notation 
 and reviews fundamental concepts 
 related to MCST games, 
 saving games, and the Shapley value.
Section~\ref{section:NullPlayer} 
    establishes structural properties 
    of MCST‑saving games, 
    including a characterization 
    of null players and lower bounds 
    on the Shapley value. 
Section~\ref{section:MonteCarlo} 
    presents the Monte Carlo–based FPRAS 
    and analyzes its computational complexity. 
Section~\ref{section:comp}
    reports computational experiments 
    demonstrating the practical performance 
    of the proposed method. 
Section~\ref{section:conclusion}
    summarizes several results 
    obtained in this paper.


\section{Preliminaries}\label{section:pre}

This section describes notations, definitions, and 
 some known properties.
For any finite set $V$, 
    let $\ComG [V]$ denote the undirected complete graph
    with vertex set $V$ and 
    edge set $E=\left( \begin{array}{c} V \\ 2 \end{array} \right)$.
Each edge $e=\{i,j\} \in E$ is assigned a non-negative edge-weight,  
    denoted by $w(e)$, $w(i,j)$ and/or $w(j,i)$.
A pair $\langle \ComG [V], \Vec{w} \rangle$ is called
    a {\em non-negatively weighted complete graph}.
If the edge-weight satisfies  
    $\Vec{\widetilde{w}} \in \{0,1\}^E$, 
    then 
    $\langle \ComG [V], \Vec{\widetilde{w}} \rangle$
    is a {\em 0-1 weighted complete graph}.

A graph $G=(V,E)$ is called a {\em tree}
    if it is connected and contains no cycles.
For a graph $G=(V, E)$, 
    a subgraph $G'=(V', E')$ of $G$
    is called a {\em spanning tree} if 
    $V=V'$ and $G'$ is a tree.
A spanning tree is identified by its set of edges.
For any edge subset $E'$ 
    of a non-negatively weighted complete graph 
    $\langle \ComG [V], \Vec{w} \rangle$, 
    the sum of edge weights 
    $\sum_{e\in E'} w(e)$
    is  referred to as the {\em cost} of $E'$.
A minimum spanning tree problem defined on 
    $\langle \ComG [V], \Vec{w} \rangle$
    finds a spanning tree $E'$ of $\ComG [V]$
    that minimizes the cost $\sum_{e\in E'} w(e)$.

Let $N=\{1,2,\dots,n\}$ be a set of players. 
Fix a distinguished node $r$, referred to as the {\em root}, 
    and set $N_r=\{ r,1,2,\dots ,n\}$.
Let $\langle \ComG [N_r], \Vec{w}\rangle$ 
    be a non-negatively weighted complete graph.
For each non-empty player subset $S\subseteq N$, 
    let $c(S)$ denote the cost of a minimum spanning tree 
    of the induced subgraph $\ComG [S \cup \{r\}]$
    with respect to the edge-weight~$\Vec{w}$.
We define $c(\emptyset )=0.$
The {\em minimum cost spanning tree game} (MCST game)
    is the cooperative (cost) game given by the pair $(N,c)$.
If the edge-weights are restricted
    to the set $\{0,1\}$, 
    then the resulting game $(N,c)$ is called 
    the {\em simple} minimum cost spanning tree game 
    (simple MCST game). 

Given a cooperative cost game, 
    the associated {\em saving game} transforms 
    this cost perspective 
    into a value perspective by measuring how much cost 
    the coalition saves compared to acting individually.
Formally, we define the saving game $(N,v)$ 
    by the characteristic function
    $v(S)=\sum_{i\in S} c(\{i\}) - c(S)$, 
    with $v(\emptyset )=0.$
In this paper, we study the saving game 
    associated 
    with the minimum cost spanning tree game $(N,c)$. 
This game, denoted by $(N,v)$, 
    will be referred to as the {\em MCST‑saving game}, 
    which was discussed in a seminal paper by Bird~\cite{bird1976cost}. 
When the underlying game is a simple MCST game, 
    the corresponding saving game
    will be called the simple MCST‑saving game.

The definition of the saving game $(N,v)$ directly implies that
    $v(\{j\})=c(\{j\})-c(\{j\})=0$ $(\forall j \in N)$.
The MCST-saving game satisfies that
    $v(S)=\sum_{j \in S} w(r,j)-c(S)$ 
    $(\emptyset \neq \forall S \subseteq N)$.
To begin with, we note that the MCST-saving game is 
    nonnegative and super-additive (see~\cite{bird1976cost}).
For any non-empty set $S\subseteq N,$
    the edge set $\{\{r,i\} \mid i \in S\}$ forms 
    a spanning tree on $\ComG [S \cup \{r\}].$ 
Thus, $v(S)=\sum_{i \in S} w(r,i) -c(S) \geq 0.$
Moreover, for any disjoint pair of coalitions $S,T \subseteq  N$,  
    we have
\[
    v(S \cup T) = \sum_{i\in S \cup T} c({i}) - c(S \cup T)
             \geq \sum_{i \in S} c({i}) - c(S) + \sum_{i \in T}c({i}) - c(T)
             = v(S) + v(T),
\]   
    since the cost function $c$ is sub-additive 
    with respect to spanning trees.
The non-negativity together with  super-additivity implies 
    the monotonicity of $v$.
    We see that, if $N \supseteq T \supseteq S$, then  
\[
    v(T)=v(T \setminus S \cup S)\geq v(T \setminus S) + v(S) \geq v(S). 
\]


The Shapley value, formally defined below, 
    is a payoff vector introduced 
    by Shapley~\cite{shapley1953value}.
A {\em permutation}  $\pi$  of players in $N$ 
 is a bijection 
    $\pi: \{1,2,\ldots ,n\} \rightarrow N$, 
    and we write  
   $\pi (j)$ for the player in position $j$ in the permutation $\pi$.
Let $\Pi_N$ be the set of all permutations defined on $N$.
The Shapley value of characteristic function form game 
    $(N,v)$ is a payoff vector    
    $\phi=(\phi_1, \phi_2, \ldots , \phi_n)$ defined by
\[
  \phi_i=\frac{1}{n!} \sum_{\pi \in  \Pi_N} \Merg (\pi, i) 
    \;\; (\forall i \in N)
\]
where $\Merg(\pi, i)=v(\Prec (\pi,i) \cup \{i\})-v(\Prec (\pi,i))$, 
 called the marginal contribution of player $i$ 
    in permutation $\pi$, 
    where $\Prec (\pi, i)$ is the set of players in $N$ 
    which precede $i$ in the permutation $\pi$. 

The monotonicity of MCST-saving games imply that 
    the corresponding Shapley value is a non-negative vector.
Next, we identify when 
    the Shapley value is equal to zero.
A player $i\in N$ is
    called a {\em null player} in the game $(N, v)$
    whenever $v(S \cup \{i \})=v(S)$ 
    for all $S \subseteq N \setminus \{i\}.$

\begin{Lemma} \label{Lemma:Null=Shap0}
In an MCST-saving game, 
    a player $i \in N$ is a null player
    if and only if his/her Shapley value is equal to zero, 
    i.e. $\phi_i=0.$
\end{Lemma}


\begin{proof}
In any cooperative game,
    if $i$ is a null player, 
    then by the definition of
    the Shapley value we have $\phi_i = 0$.  
Conversely, in an MCST--saving game, 
    if $\phi_i = 0$, 
    then by the monotonicity of $v$,  
    player $i$ cannot make a positive marginal contribution to any coalition.  
Hence $i$ must be a null player.
\end{proof}

A player $i \in N$ is a null player 
    of an MCST-saving game $(N, v)$
    if and only if 
\[
\begin{aligned}
 c(S\cup \{i\})  &=
     -v(S\cup \{i\})+\sum_{j\in S}w(r,j)+w(r,i)
    =-v(S)  +\sum_{j\in S}w(r,j)+w(r,i) \\
  &= c(S)+w(r,i) \;\;\;\; 
  (\forall S \subseteq N\setminus \{i\}),
\end{aligned}
\]
that is,
    a minimum spanning tree 
    of $\ComG [S\cup \{r, i\}]$
    is obtained by adding edge $\{r,i\}$
    to any minimum spanning tree
    of $\ComG [S \cup \{r\}]$
    for all $S \subseteq N\setminus \{i\}.$

In contrast to the saving game, 
    MCST games fail to satisfy the property stated in Lemma~\ref{Lemma:Null=Shap0}. 
To illustrate this point, 
    consider the MCST game $(N,c)$ defined by
    $\langle \ComG[N_r], \Vec{w} \rangle$ 
    with $N_r=\{r,1,2\}$ and edge weights 
    $w(r,1)=1,\; w(1,2)=2,\; w(r,2)=4.$ 
Since $c(\{1\})=1,\; c(\{2\})=4,\; c(\{1,2\})=3,$ 
    the Shapley value of $(N, c)$ is  
    $(\varphi_1,\varphi_2)=(0,3).$ 
Yet player~1 is clearly not a null player of $(N, c)$. 
This observation underscores the intractability of the problem 
    of calculating the Shapley value in MCST games 
    and thus provides the motivation 
    for introducing the saving game, 
    which subsequently allows the construction of an FPRAS.
Let $(N, v)$ denote the saving game corresponding 
    to the characteristic function form game introduced above.
It is given by
    $v(\{1\})=v(\{2\})=0$ and $v(\{1,2\})=(1+4)-3=2.$
The Shapley value of $(N, v)$ is $(\phi_1, \phi_2)=(1,1)$.
Moreover, this payoff vectors admits the decomposition 
    $(\phi_1, \phi_2)=(1,1)=(1,4)-(0,3)=(w(r,1),w(r,2))-(\varphi_1, \varphi_2).$

We now describe, without proof,  
    the relationships between a game 
    in characteristic function form 
    and its corresponding saving game.
\begin{remark}
    
Let $(N,v)$ be a saving game 
    of a characteristic function form game $(N,c).$

\begin{enumerate}
\item  A player $i$ is a null player in $(N, v)$
    if and only if $i$ is a dummy player in $(N, c)$,
    that is 
\[
    c(S \cup \{i\})=c(\{i\})+ c(S) \;\;
        (\forall S \subseteq N \setminus \{i\}).
\]

\item
The Shapley value  $\phi$ of $(N, v)$ satisfies 
$\phi_i+\varphi_i=c(\{i\}) \;\; (\forall i \in N), $
where $\varphi$ denotes the Shapley value of $(N, c)$.
\end{enumerate}
\end{remark}

\section{Null Players in MCST-Saving Games} 
\label{section:NullPlayer}

In this section, we establish 
    necessary and sufficient conditions 
    for a player to be a null player 
    in an MCST-saving game.
This result is crucial for constructing 
    the FPRAS in the later section.

\begin{Theorem}\label{Theorem:Null-Player}
    Let $(N,v)$ be an MCST-saving game defined 
    by the non-negatively weighted complete graph
    $\langle \ComG [N_r], \Vec{w} \rangle$.
For any player $i \in N$, 
    the following statements are equivalent:
\begin{description}
\item[\rm (0)] the Shapley value of player $i$ is zero; 
\item[\rm (1)] player $i$ is a null player
    {\rm (i.e., $\forall S \subseteq N\setminus \{i\}, v(S \cup \{i\})=v(S)$)}; 
\item[\rm (2)] $\forall j \in N\setminus \{i\}$, $v(\{i,j\})=v(\{j\})$;
\item[\rm (3)] $\forall j \in N\setminus \{i\}$,  
    $w(r,i)\leq w(i,j) \geq w(r,j)$.
\end{description}
\end{Theorem}

\begin{proof} 
(0)$\iff$(1): See Lemma~\ref{Lemma:Null=Shap0}. 

\noindent 
(1)$\implies$(2):
This follows immediately from the case $|S|=1$.

\noindent 
(2)$\implies$(3):
We prove the contrapositive. 
Suppose that there exists
    $j \in N \setminus \{i\}$ satisfying that 
    [$w(r,i)> w(i,j)$ or $w(i,j) < w(r,j)$].
This condition is equivalent to
     $w(i,j) < \max \{w(r,i), w(r,j)\}$.
It then follows that
\begin{eqnarray*}
c(\{i\})+c(\{j\}) 
    &=& w(r,i)+w(r,j) 
    =\max \{w(r,i),w(r,j)\} + \min \{w(r,i),w(r,j)\} \\
    &>& w(i,j) + \min \{w(r,i),w(r,j)\} 
    = \min \{w(i,j)+w(r,i),w(i,j)+w(r,j)\} \\
    &\geq & c(\{i,j\}).
\end{eqnarray*}
Consequently, we obtain 
$
    v(\{j\})=0<   c(\{i\})+c(\{j\}) -   c(\{i,j\}) =v(\{i,j\}).
$

\smallskip 

\noindent
(3)$\implies$(1):
We show that
    if the player $i \in N$ is not a null player, 
    then $\exists j \in N \setminus \{i\}$ satisfying 
    [$w(r,i)> w(i,j)$ or $w(i,j) < w(r,j)$].
As $i$ is not a null player, 
    there exists a subset  
    $S \subseteq N \setminus \{i\}$ satisfying  
    $v(S)\neq v(S \cup \{i\})$.
Let $T_S$ be the set of edges of a minimum spanning tree 
    of $\ComG [S \cup \{r\}]$ 
    with respect to the given edge-weight $\Vec{w}$.    
The assumption and monotonicity of $v$ implies that
\begin{eqnarray*}
    0&<& v(S \cup \{i\})- v(S) =( w(r,i)+c(S)) -c(S \cup \{i\})
    =w(r,i) + \sum_{e \in T_s} w(e) -c(S \cup \{i\})
\end{eqnarray*}
    and thus the spanning tree $T_S \cup \{\{r,i\}\}$ 
    of $\ComG [S \cup \{r,i\}]$
    is not the minimum spanning tree.
The above non-optimality implies that
    there exists an edge $e \not \in T_S \cup \{\{r,i\}\}$
    and an edge $f$ such that 
    (1) $f$ lies on the unique path $P$ 
        in  $\ComG [S \cup \{r,i\}]$ 
        connecting the end vertices of $e$ and
    (2)  $w(e) < w(f)$.
If $e$ is an edge connecting vertices in $S \cup \{r\}$, 
    then the unique path $P$ is contained 
    in  $\ComG [S \cup \{r\}]$ and 
    $T_s \setminus \{f\} \cup \{e\}$ is a spanning tree 
    of  $\ComG [S \cup \{r\}]$, 
    whose weight is $c(S)-w(f)+w(e) < c(S)$,  
    which contradicts the optimality of $T_S$.
Thus, $e$ is incident to $i$.
The property $e \not \in T_S \cup \{\{r,i\}\}$
    implies that $e \neq \{r,i\}$.
In what follows, we write $e = \{i, j\}$, 
    thereby defining the vertex $j \in S$.
We proceed to show, 
    by the way of contradiction, 
    that $j (\in S)$ satisfies either 
    $w(r,i)> w(i,j)$ or $w(i,j) < w(r,j)$.
    
Assume on the contrary that 
    $w(r,i) \leq  w(i,j)=w(e) \geq w(r,j)$.
The inequalities $w(r,i) \leq w(i,j) = w(e) < w(f)$ implies 
    that $\{r,i\} \neq f$. 
Since $i$ is a leaf vertex of the tree 
    defined by $T_S \cup \{\{r,i\}\}$, 
    the unique path $P$ connecting $i$ and $j$ 
    includes the edge $\{r,i\}$ as the end edge.
As $f \neq \{r,i\}$, 
    the edge $f$ lies on the subpath $P'$ of $P$
    obtained by deleting the edge $\{r,i\}$ from $P$.
Clearly, the subpath $P'$ connects vertices $r$ and  $j.$
The inequalities $w(f) > w(e)=w(i,j) \geq w(r,j)$ imply that 
    $\{r,j\} \neq f \in P'$.
As $P'$ connects 
    end vertices of $\{r, j\}$, 
    $P'$ does not include the edge $\{r, j\}$.
Thus, the edge set $T_S \setminus \{f\} \cup \{r,j\}$ 
    is a spanning tree of $\ComG [S \cup \{r\}]$,
    whose weight $c(S)-w(f)+w(r,j)\leq c(S)-w(f)+w(e) < c(S)$, 
    which contradicts the optimality of $T_S$.
\end{proof}

By employing Theorem~\ref{Theorem:Null-Player}, 
    we can determine whether a given player 
    in an \mbox{MCST-saving} game
    is a null player in $O(n)$ time.
When player $i$ is a null player, 
    Lemma~\ref{Lemma:Null=Shap0} implies that 
    his/her Shapley value is equal to zero. 
The Shapley values of remaining players are obtained 
    by considering the MCST-saving game 
    derived from deleting the vertex corresponding 
    to the null player.
Thus, by applying the above procedure repeatedly, 
    we obtain an MCST-saving game in which  
    every  remaining player is non-null, within $O(n^2)$ time.

When restricted to simple MCST-saving games, 
    the above theorem takes the following form, 
    which provides a lower bound 
    for the Shapley value of a non‑null player
    in the simple MCST-saving game.

\begin{Cor} \label{Cor:ShapleyLB}
Let $(N,v)$ be a simple MCST-saving game 
    defined by the 0-1 weighted complete graph
  $\langle \ComG [N_r], \Vec{\widetilde{w}} \rangle$.
Let $\phi$ denote the corresponding Shapley value.
For any player $i \in N$, 
    the following statements are equivalent:
\begin{description}
\item[\rm (0)] $\phi_i \neq 0,$
\item[\rm (1)] player $i$ is ``not'' a null player,
\item[\rm (2)]  $\exists j \in N\setminus \{i\},$  
 $(\widetilde{w}(r,i),\widetilde{w}(i,j),\widetilde{w}(r,j))
    \in \{(1,0,0), (1,0,1), (0,0,1)\}$,
\item[\rm (3)] $\displaystyle \phi_i \geq \frac{1}{n(n-1)}$. 
\end{description}
\end{Cor}

\begin{proof}
We prove the implication \((2)\Rightarrow(3)\).  
The remaining cases follow directly 
    from Theorem~\ref{Theorem:Null-Player}.

Let $j\in N\setminus \{i\}$ be a player 
   such that 
    $(\widetilde{w}(0,i),\widetilde{w}(i,j),\widetilde{w}(0,j))
    \in \{(1,0,0), (1,0,1), (0,0,1)\}$.
In each of these cases, the game \((N,v)\) satisfies
    $v(\{i\})=v(\{j\})=0$ and $v(\{i,j\})=1$.
Define the set of permutations 
   \mbox{$\Pi_N'=\{\pi \in \Pi_N \mid \pi(1)=j, \pi(2)=i\}$.}
Then, the Shapley value $\phi_i$ 
    satisfies 
\begin{eqnarray*}
    \phi_i&=&\frac{1}{n!}  
        \sum_{\pi \in \Pi_N} \Big( v(\Prec(\pi, i)\cup \{i\}) - v(\Prec(\pi, i) ) \Big) \\
    &\geq& \frac{1}{n!}
            \sum_{\pi \in \Pi_N'} \Big(v(\{i,j\}) - v(\{j\}) \Big)   
    = \frac{\sum_{\pi \in \Pi_N'} (1-0)  
        }{n!}
    = \frac{(n-2)!}{n!}
    = \frac{1}{n(n-1)},
\end{eqnarray*}
where the above inequality follows from 
    the monotonicity of $(N, v)$.
\end{proof}


\section{Monte Carlo Method 
}\label{section:MonteCarlo}

In this section, we introduce a Monte Carlo method 
    for computing the Shapley value of MCST-saving games,
    which constitutes 
    a fully polynomial-time randomized approximation scheme
    (FPRAS). 
We now describe a Monte Carlo method, 
    which corresponds to the most simple form 
    of the algorithms 
    for calculating the Shapley-Shubik index 
    proposed by Mann and Shapley in their seminal work~\cite{Mann1960RM2651}.
For an overview of Monte Carlo methods 
    in Shapley value (and/or Shapley-Shubik index)
    computation, 
    see~\cite{bachrach2010approximating,
    Liben-Nowell2012Computing,
    ushioda2022monte} for example.

As shown in the previous section, 
    we can determine the set of all the non-null players within $O(n^2)$ time.
Throughout this section, we assume the following.
\begin{Assumption} \label{Assumption:non-null}
    Every player in the given MCST-saving game is non-null, 
    and hence has a positive Shapley value. 
\end{Assumption}

\begin{algorithm}
\caption{}\label{algo1}
\begin{algorithmic}[Algorithm 1]
\Input  
     non-negatively weighted complete graph $\langle \ComG [N_r], \Vec{w}\rangle$
     and a positive integer $M$.
\Output Approximate Shapley value $\phi^A_i$ for the MSCT-saving game $(N,v)$ 
    defined by the non-negatively weighted complete graph  $\langle \ComG [N_r], \Vec{w}\rangle$.
\smallskip
\State Set $m\gets0$ and $\phi'_i\gets0
  $ $(\forall i \in N)$.
\While{$m < M$}
\State  Choose $\pi \in \Pi_N$ randomly.
\State \label{step1.2} Update $\phi'_i \gets \phi'_i + \Merg (\pi ,i )$ $(\forall i \in N)$ and $m \gets m + 1$.
\EndWhile
\State Output $\phi^A_i = \frac{1}{M}\phi'_i$ for each $i \in N$.
\end{algorithmic}
\end{algorithm}

Algorithm~\ref{algo1} denotes a simple Monte Carlo method.
We denote the vector (of random variables) 
	obtained by executing Algorithm~\ref{algo1}
	by 	
	$(\phi^A_1, \phi^A_2,\ldots ,
		\phi^A_n)$.
The following result is immediate.

\begin{Lemma} \label{Lemma:Unbiasedness}
    For each player $i \in N$, 
	$\Exp \left[ \phi^A_i \right]=\phi_i$.
\end{Lemma}

\begin{proof}
    For each $m \in \{1, 2, \ldots, M\}$ and $i \in N$, 
    define the random variable 
    $\Delta^{(m)}_i = \Merg(\pi, i)$, 
    where $\pi \in \Pi_N$ is 
    the permutation sampled 
    in the $m$-th iteration of Algorithm~\ref{algo1}.
As Algorithm~\ref{algo1} chooses a permutation $\pi \in \Pi_N$ randomly, 
    it is obvious that for each player $i \in N$, 
	$\{ \Delta^{(1)}_i, \Delta^{(2)}_i, \ldots , \Delta^{(M)}_i\}$ 
    is an i.i.d. sequence satisfying
	$\phi^A_i=\sum_{m=1}^M \Delta^{(m)}_i/M$ and 
 	$ \Exp \left[ \phi^A_i \right]
	=\Exp \left[ \Delta^{(m)}_i \right]=\phi_i$ 
	$(\forall m \in \{1,2,\ldots ,M\})$.
\end{proof}

Next, we show that 
    the time complexity of Algorithm~\ref{algo1} is $O(Mn^2)$.
In the following, we discuss a problem of calculating 
    $(c(\{\pi(1)\}),c(\{\pi(1), \pi(2)\}),\ldots , c(N))$
    for a given permutation $\pi \in \Pi_N$.
The \mbox{0-1} weighted complete graph case 
    is discussed in~\cite{ando2020monte,tokutake2013approximaiton}. 
For each player $i \in N$, 
    it holds that $c(\{i\})=w(r,i)$.
Let  $(S, i)$ be a pair consisting of
    a non-empty coalition $S$ and 
    a player $i \in N \setminus S$.
Suppose that we already know the value $c(S)$ 
    and the set of edges, denoted by $T_S$, 
    forming a minimum spanning tree 
    of $\ComG [S \cup \{r\}]$ 
    with respect to the given edge-weight $\Vec{w}$.   
We now consider the problem of computing the value 
    $c(S \cup \{i\})$ and the set of edges 
    corresponding to a minimum spanning tree 
    of $\ComG[S \cup \{r,i\}]$ 
    that realizes this value. 
Let $\overline{T_S}$ denote the set of edges 
    of $\ComG[S \cup \{r\}]$ that are not contained in $T_S$.
It is well known that for any edge $e\in \overline{T_S}$,  
    the unique path in $T_S$ connecting the endpoints of $e$, 
    denoted by $P_e$, 
    satisfies $w(e) \ge w(f)$ for every $f \in P_e$. 
We now show that there exists a minimum spanning tree
    of $\ComG[S \cup \{r,i\}]$, which, in particular, 
    contains no edge from $\overline{T_S}$. 
Let $T^+$ denote a minimum spanning tree
    of $\ComG[S \cup \{r,i\}]$, 
    and assume that 
    there exists an edge 
    $e \in T^+ \cap \overline{T_S}$. 
Then there exists an edge $f \in P_e$ such that
    $T^+ \setminus \{e\} \cup \{f\}$
    is also a spanning tree of $\ComG[S \cup \{r,i\}].$
The minimality of $T^+$ together with
    $w(e)\geq w(f)$ implies that 
     $T^+ \setminus \{e\} \cup \{f\}$
     is likewise a minimum spanning tree 
     of $\ComG[S \cup \{r,i\}]$.
The above exchange operation decreases 
    the number of edges in
    $T^+ \cap \overline{T_S}$ by one. 
By repeatedly applying this operation, 
    we eventually obtain a minimum spanning tree 
    of $\ComG[S \cup \{r,i\}]$, 
    which is edge‑disjoint from $\overline{T_S}$.
Therefore, we only need to find a minimum spanning tree 
    in the undirected graph $G'$ 
    with vertex set $S \cup \{r,i\}$ 
    and edge set 
\[
\left( \begin{array}{c} 
    S \cup \{r,i\} \\ 2
        \end{array}
\right) \setminus \overline{T_S}
= T_S \cup \{\{i,j\} \mid j \in S \cup \{r\}\}
\]
    with respect to the given edge weight $\Vec{w}$. 
As $T_S$ is a spanning tree of $\ComG[S \cup \{r\}]$
    and $\{\{i,j\} \mid j \in S \cup \{r\}\}$ 
    is a star tree of $\ComG[S \cup \{r,i\}]$, 
    it follows immediately that $G'$ is planar.
(Since every cycle in $G'$ passes the mutual vertex $i$
    and thus
    the well-known Kuratowski’s theorem~\cite{kuratowski1930} 
    says that $G'$ is planar.)
We can find a minimum spanning tree of planar graph 
    in linear time~\cite{cheriton1976,matsui1995minimum}.
Thus, for any permutation $\pi \in \Pi_N$, 
    we can calculate the vector 
    $(c(\{\pi(1)\}),c(\{\pi(1), \pi(2)\}),\ldots , c(N))$
    in $O(n^2)$ time.
Hence, the total time complexity of Algorithm~\ref{algo1} 
    is bounded by $O(Mn^2)$.


\subsection{0-1 Weighted Complete Graph}

In this subsection, we assume that 
    0-1 weighted complete graph $\langle \ComG [N_r], \widetilde{w} \rangle$ is given.
We discuss the simple MCST-saving game $(N,v)$ corresponding to the simple MCST game $(N,c)$
    defined by $\langle \ComG [N_r], \widetilde{w} \rangle$.
    
\begin{Lemma} \label{Lemma:LB-UB}
For any permutation $\pi \in \Pi_N$ and a player $i \in N$, 
    $\Merg (\pi,i)\in [0,n-1]$.
\end{Lemma}

\begin{proof}
The monotonicity of $v$ implies that $\Merg (\pi ,i ) \geq 0$.
As the given graph is 0-1 weighted,  $c(\Prec(\pi, i)) \leq n - 1$.
In case that $c(\Prec(\pi, i)\cup \{i\}) \ge 1$, we have
\[
\Merg (\pi ,i ) = \widetilde{w}(r,i)-c(\Prec(\pi, i)\cup \{i\})+c(\Prec(\pi, i))
    \leq 1 - 1 + (n - 1) = n - 1.
\]
If $c(\Prec(\pi, i)) = n - 1$ and $c(\Prec(\pi, i)\cup \{i\}) = 0$, 
    then it follows that the edge weight $\widetilde{w}(r,i) = 0$.
Therefore, we always have $\Merg (\pi ,i ) \le n - 1$.
\end{proof}

\begin{Theorem}\label{Theorem:FPRAS-01}
Let $\phi^A=(\phi^A_1, \phi^A_2,\ldots , \phi^A_n)$
    denote the vector obtained by applying 
    Algorithm~\ref{algo1}
	to a simple MCST-saving game $(N,v)$. 
Then, for any $\varepsilon >0$ and  $0< \delta <1$,
	we have the following.

\noindent 
{\rm (1)}
If we set 
	$\displaystyle M \geq \frac{n^2 (n-1)^4 \ln (2/\delta)}{2\varepsilon^2}$,
	then each player $i \in N$ satisfies that
\[
\Prob 
	\left[ 
		\frac{| \phi^A_i-\phi_i| }{\phi_i}
            < \varepsilon  
	\right] \geq 1-\delta.
\]

\noindent 
{\rm (2)}
 If we set 	
	$\displaystyle M\geq \frac{n^2 (n-1)^4 \ln (2n/\delta)}{2\varepsilon^2}$,
	then \
\[
 \Prob 
	\left[
		\forall i \in N, 
             \frac{| \phi^A_i-\phi_i| }{\phi_i}
            < \varepsilon  
	\right] \geq 1-\delta.
\]
\end{Theorem}

\begin{proof}
By Hoeffding's inequality~\cite{hoeffding1963probability} 
    and Lemma~\ref{Lemma:LB-UB},   
	each player $i \in N$ satisfies
\begin{eqnarray*}
\Prob \left[ \left| \phi^A_i
	- \Exp \left[ \phi^A_i \right] \right|
	\geq t \right]
&\leq 
& 2 \exp \left(
	- \frac{2M^2 t^2}{\sum_{m=1}^M ({(n-1)}-0)^2}
	\right)
= 2 \exp \left( \frac{-2M t^2}{(n-1)^2} \right).
\end{eqnarray*}

As the graph has $0$–$1$ weighted edges, 
    Assumption~\ref{Assumption:non-null} together with 
    Corollary~\ref{Cor:ShapleyLB} ensures that 
    $\phi_i \geq \frac{1}{n(n-1)}$ for all $i\in N$,
    so the denominator of
    $\frac{\lvert \phi_i^A - \phi_i \rvert}{\phi_i}$
    is well defined.

\noindent
{\rm (1)} If we set 
	$\displaystyle M \geq \frac{n^2(n-1)^4 \ln (2/\delta)}{2\varepsilon^2}$,
	then 
\[
\begin{array}{l}
\displaystyle 
\Prob 
	\left[ 
		\frac{\left| \phi^A_i-\phi_i \right| }{\phi_i}
            < \varepsilon  
	\right] 
=1- 
\Prob 
	\left[ 
		\left| \phi^A_i-\Exp \left[ \phi^A_i \right] \right| 
            \geq \varepsilon  \phi_i
	\right] 
\geq 1- \Prob 
	\left[ 
		\left| \phi^A_i-\Exp \left[ \phi^A_i \right] \right| 
            \geq  \left( \frac{{\varepsilon }}{n(n-1)} \right) 
	\right] \\
\displaystyle 
\geq 1- 2 \exp 
    \left( 
        \frac{-2M \left(\frac{{\varepsilon}}{n(n-1)}\right)^2
        }{(n-1)^2} 
    \right)
=    1- 2 \exp \left( -M \frac{2  \varepsilon^2}{n^2(n-1)^4} \right) \\
\geq 1-  2 \exp \left( - \ln \left( \frac{2}{\delta} \right) \right) 
 = 1- \delta.
 \end{array}
\] 

\noindent
{\rm (2)} If we set 
	$\displaystyle M\geq \frac{n^2(n-1)^4 \ln (2n/\delta)}{2\varepsilon^2}$,
	then we have that
\[
\begin{array}{l}
\displaystyle 
\Prob \left[
				\forall i \in N, 
					\frac{\left| 
						\phi^A_i	- \phi_i 
					\right|}{\phi_i}
					< \varepsilon 
			\right]
 = 1- \Prob 
			\left[
				\exists i \in N, 
					\left|
                        \phi^A_i-\Exp \left[ \phi^A_i \right] 
                    \right| 
					\geq \varepsilon \phi_i
			\right] \\
\displaystyle
 \geq 
 1-\sum_{i\in N} 
	\Prob 
			\left[
					\left|
                        \phi^A_i-\Exp \left[ \phi^A_i \right] 
                    \right| 
					\geq \varepsilon \phi_i
			\right]
\geq 
1-\sum_{i \in N} 
	\Prob 
			\left[
					\left|
                        \phi^A_i-\Exp \left[ \phi^A_i \right] 
                    \right| 
					\geq  \left( \frac{\varepsilon}{n(n-1)} \right)
			\right] \\
\displaystyle 
\geq 1- \sum_{i\in N}  2 \exp 
        \left( 
            \frac{-2M 
                \left( 
                    \frac{{\varepsilon}}{n(n-1)}
                \right)^2 
            }{{(n-1)}^2}
        \right)
\geq 1- \sum_{i \in N}  2 \exp 
        \left( -M 
            \frac{2 \varepsilon^2 }{n^2(n-1)^4}
        \right) \\
\displaystyle
 \geq 
 1- \sum_{i \in N}  2 \exp 
	\left( 
		- \ln \left( \frac{2n}{\delta} \right) 
 	\right) 
= 1-\sum_{i \in N} \frac{\delta}{n}=1-\delta.  
\end{array}
\]
\end{proof}

\subsection{Non-negatively Weighted Complete Graph}

In this subsection, 
  a non-negatively weighted complete graph 
  $\langle \ComG [N_r], w \rangle$ 
  is assumed to be given.
We consider the MCST-saving game $(N,v)$ 
    corresponding to the MCST game $(N,c)$ 
    defined by $\langle \ComG [N_r], w \rangle$.

Consider a set of mutually distinct positive edge-weights
\[
    \{\gamma_1, \gamma_2, \ldots , \gamma_H\}
    =\{\gamma \in \mathbb{R} 
    \mid \exists \{i,j\}\subseteq N_r, \gamma=w(i,j)>0 \}
\]
    such that  $0< \gamma_1 < \cdots < \gamma_H$.
Obviously we have $H\leq n(n+1)/2.$
Define $\gamma_0 =0$.
For each $h\in \{1,2,\ldots ,H\},$
    we introduce a 0-1 edge-weight
\[
    w^{[h]} (i,j)=\left\{ 
        \begin{array}{ll}
            1 & (\mbox{if } \gamma_h \leq w(i,j)), \\
            0 & (\mbox{otherwise}),
        \end{array}
        \right.
    \quad
    \mbox{for each edge } \{i,j\} \subseteq N_r. 
\]
Obviously, we have 
    $w(i,j)=\sum_{h=1}^H (\gamma_h - \gamma_{h-1} )
        w^{[h]} (i,j)$
    for each edge $\{i,j\} \subseteq N_r.$
Let $(N, c^{[h]})$ denote 
    the simple MCST game defined 
    by a 0-1 weighted complete graph 
    $\langle \ComG [N_r], w^{[h]} \rangle$
    for each $h \in \{1,2,\ldots ,H\}$.
Then, Norde, Moretti and Tijs~\cite{norde2004minimum} showed 
    the following decomposition property of MCST game.
\begin{Theorem} \label{Theorem:0-1Decomp}
    The MCST game $(N,c)$ defined by 
    non-negatively weighted complete graph
    $\langle \ComG [N_r], \Vec{w} \rangle$
    can be expressed as
    $c=\sum_{h=1}^H (\gamma_h - \gamma_{h-1} )
        c^{[h]}$.
    Moreover, the Shapley value $\varphi$ of $(N,c)$ is given by
      $\varphi=\sum_{h=1}^H (\gamma_h - \gamma_{h-1} )
        \varphi^{[h]}$, 
        where  $\varphi^{[h]}$ is the Shapley value 
        of  $(N, c^{[h]})$.
\end{Theorem}

\noindent
A concise proof is presented by Ando in~\cite{ando2012computation}.

For each $h \in \{1,\ldots ,H\}$, 
    let $(N, v^{[h]})$ denote the saving game 
    of $(N, c^{[h]})$ 
    defined by 
\[
    v^{[h]}(S)= \sum_{i \in S} w^{[h]} (r,i) 
        - c^{[h]} (S) \;\;
        (\forall S \subseteq N).
\]
Then, we have the following properties.

\begin{Cor} \label{Cor:decomp}
    For the MCST-saving game $(N,v)$ 
    defined by the MCST game $(N,c)$,
    we have
    $v=\sum_{h=1}^H (\gamma_h - \gamma_{h-1} )
        v^{[h]}$.
    In addition, the Shapley value 
    $\phi$ of $(N,v)$ satisfies that 
      $\phi=\sum_{h=1}^H (\gamma_h - \gamma_{h-1} )
        \phi^{[h]}$, 
        with $\phi^{[h]}$ denoting 
        the Shapley value 
        of  $(N, v^{[h]})$.
\end{Cor}

\begin{proof}
    The MCST-saving game $(N, v)$ satisfies that
\begin{eqnarray*}
        v(S)
        &=&\sum_{i\in S}w(r,i)-c(S)
        =\sum_{i\in S} 
        \sum_{h=1}^H (\gamma_h - \gamma_{h-1} )
        w^{[h]} (r,i)
        - \sum_{h=1}^H (\gamma_h - \gamma_{h-1} )
        c^{[h]}(S) \\
        &=&  \sum_{h=1}^H (\gamma_h - \gamma_{h-1} ) \left(
        \sum_{i\in S} w^{[h]} (r,i)
        - c^{[h]}(S) 
        \right)
        = \sum_{h=1}^H (\gamma_h - \gamma_{h-1} ) v^{[h]}(S). 
\end{eqnarray*}
The additivity 
    of the Shapley value implies that 
      $\phi=\sum_{h=1}^H (\gamma_h - \gamma_{h-1} )
        \phi^{[h]}$.
\end{proof}

In the following, we describe Algorithm~\ref{algo2}, 
  which is included only for analytical purposes.
The output of this algorithm coincides 
    exactly with that of Algorithm~\ref{algo1}.
In Algorithm~\ref{algo2}, we define that
\[
  \Merg^{[h]} (\pi ,i )
    =v^{[h]}(\Prec (\pi,i) \cup \{i\})-v^{[h]}(\Prec (\pi,i))
\]
    for each $h\in \{1,2,\ldots H\}$
    and $\forall (\pi, i)\in \Pi_N \times N$.

\begin{algorithm}
\caption{}\label{algo2}
\begin{algorithmic}[Algorithm 2]
\Input  
     non-negatively weighted complete graph $\langle \ComG [N_r], \Vec{w}\rangle$
     and a positive integer $M$.
\Output Approximate Shapley value $\phi^A_i$ for the MSCT-saving game $(N,v)$ 
    defined by the non-negatively weighted complete graph  $\langle \ComG [N_r], \Vec{w}\rangle$.
\State  Set  $\phi'^{[h]}_i\gets0$
    $(\forall h\in \{1,2,\ldots ,H\}, \forall i \in N)$.
\State Set $m\gets0$ and $\phi'_i\gets0$
        $(\forall i \in N)$.
\While{$m < M$}
\State  Choose $\pi \in \Pi_N$ randomly.
\State Update 
    $\phi'^{[h]}_i \gets \phi'^{[h]}_i + \Merg^{[h]} (\pi ,i )$ 
    $(\forall h \in \{1,2,\ldots ,H\}, \forall i \in N)$.
\State  Update 
    $\phi'_i \gets \phi'_i + \Merg (\pi ,i )$ 
    $(\forall i \in N)$ and $m \gets m + 1$.
\EndWhile
\State Set  $\phi^{[A][h]}_i=\frac{1}{M}\phi'^{[h]}_i$ 
    $(\forall h \in \{1,2,\ldots ,H\}, \forall i \in N)$.
\State Output $\phi^A_i = \frac{1}{M}\phi'_i$ for each $i \in N$.
\end{algorithmic}
\end{algorithm}

\begin{Theorem}\label{Theorem:FPRAS-non-neg}
Let $\phi^A=(\phi^A_1, \phi^A_2,\ldots , \phi^A_n)$
    denote the vector obtained by applying 
    Algorithm~\ref{algo1} and/or Algorithm~\ref{algo2}
	to a MCST-saving game $(N,v)$
    defined by a non-negatively weighted complete graph.
For any $\varepsilon >0$ and  $0< \delta <1$,
	we have the following.

\noindent 
{\rm (1)}
If we set 
	$\displaystyle M \geq \frac{n^2(n-1)^4 \ln (2H/\delta)}{2\varepsilon^2}$,
	then each player $i \in N$ satisfies that
\[
\Prob 
	\left[ 
		\frac{| \phi^A_i-\phi_i| }{\phi_i}
            < \varepsilon  
	\right] \geq 1-\delta.
\]

\noindent 
{\rm (2)}
 If we set 	
	$\displaystyle M\geq \frac{n^2(n-1)^4 \ln (2nH/\delta)}{2\varepsilon^2}$,
	then 
\[
 \Prob 
	\left[
		\forall i \in N, 
             \frac{| \phi^A_i-\phi_i| }{\phi_i}
            < \varepsilon  
	\right] \geq 1-\delta.
\]
\end{Theorem}

\begin{proof}
For any $h \in \{1,2,\ldots , H\}$, 
    $(N,v^{[h]})$ is a simple MCST-saving game and thus, 
    Lemma~\ref{Lemma:LB-UB} implies that
    $\Merg^{[h]} (\pi, i) \in [0,n-1]$ 
    $(\forall (\pi, i) \in \Pi_N \times N)$.
    
Let $\phi^{[h]}$ denotes the Shapley value of
    $(N, v^{[h]})$ for each $h \in \{1,2, \ldots , H\}$.
Lemma~\ref{Lemma:Unbiasedness} implies that 
	$\Exp \left[ \phi^{[A][h]}_i \right]=\phi^{[h]}_i$ 
    $(\forall h \in \{1,2,\ldots , H\},  \forall i \in N)$.
We define 
$
    {\cal H}_i=\{h \in \{1,2,\ldots ,H\} \mid \phi^{[h]}_i \neq 0 \}.
$
Corollary~\ref{Cor:decomp} and Assumption~\ref{Assumption:non-null} imply that 
     ${\cal H}_i \neq \emptyset$ $(\forall i \in N)$.
When $h \not \in {\cal H}_i$, 
    Lemma~\ref{Lemma:Null=Shap0} implies that
    $\Merg^{[h]}(\pi, i)=0$ $(\forall \pi \in \Pi_N)$ and thus 
    $\phi^{[A][h]}_i =\phi^{[h]}=0$.
If $h \in {\cal H}_i$, 
    then Corollary~\ref{Cor:ShapleyLB} implies that 
    $\phi^{[h]}_i \geq \frac{1}{n(n-1)}$.

The definition of Algorithm~\ref{algo2} 
    directly implies that, 
    for every player $i \in N$,  
\begin{eqnarray*}
    \phi^A_i &=&
        \left(\frac{1}{M}\right) 
        \sum_{m=1}^M \Merg (\pi^{(m)}, i) 
    = \left(\frac{1}{M}\right) 
        \sum_{m=1}^M  
        \bigl(
            v(\Prec (\pi^{(m)}, i) \cup \{i\})- v(\Prec(\pi^{(m)}, i))
        \bigr) \\
    &=& \left(\frac{1}{M}\right) 
        \sum_{m=1}^M   \sum_{h=1}^H (\gamma_h -\gamma_{h-1})
        \left(  v^{[h]}(\Prec (\pi^{(m)}, i) \cup \{i\})   
            - v^{[h]}(\Prec (\pi^{(m)}, i)) \right) \\
    &=& \left(\frac{1}{M}\right) \sum_{m=1}^M \sum_{h=1}^H 
        (\gamma_h - \gamma_{h-1}) \Merg^{[h]} (\pi^{(m)}, i) \\
    &=& \sum_{h=1}^H
        \left(
            (\gamma_h - \gamma_{h-1})  
             \left(\frac{1}{M}\right)
             \sum_{m=1}^M\Merg^{[h]} (\pi^{(m)}, i) 
        \right) 
    =  \sum_{h=1}^H (\gamma_h - \gamma_{h-1}) \phi^{[A][h]},
\end{eqnarray*}
\noindent 
    where $\pi^{(m)}$ denotes the permutation 
    sampled in the $m$-th iteration of Algorithm~\ref{algo2}.

\medskip 
\noindent
(1) If we set 
	$\displaystyle M \geq \frac{n^2(n-1)^4 \ln (2H/\delta)}{2\varepsilon^2}$,
	then 
\[
\begin{array}{l}
\displaystyle 
\Prob \left[
					\frac{\left| 
						\phi^A_i	- \phi_i 
					\right|}{\phi_i}
					< \varepsilon 
			\right]
=
\Prob 
	\left[ 
		\frac{\left| 
            \sum_{h=1}^{H} (\gamma_h - \gamma_{h-1}) \phi_i^{[A][h]}
            -\sum_{h=1}^{H} (\gamma_h - \gamma_{h-1})\phi_i^{[h]}
            \right|
            }{
            \sum_{h=1}^{H} (\gamma_h - \gamma_{h-1})\phi_i^{[h]}}
            < \varepsilon  
	\right] \\
\displaystyle  
\geq \Prob 
	\left[ 
		\frac{
            \sum_{h=1}^H
            (\gamma_h - \gamma_{h-1})
            \left| 
                \phi_i^{[A][h]}-\phi_i^{[h]}
            \right|
            }{
            \sum_{h=1}^H (\gamma_h - \gamma_{h-1})\phi_i^{[h]}}
            < \varepsilon  
	\right] 
= \Prob 
	\left[ 
		\frac{
            \sum_{h\in {\cal H}_i}
            (\gamma_h - \gamma_{h-1})
            \left| 
                \phi_i^{[A][h]}-\phi_i^{[h]}
            \right|
            }{
            \sum_{h\in {\cal H}_i} (\gamma_h - \gamma_{h-1})\phi_i^{[h]}}
            < \varepsilon  
	\right] \\
\displaystyle
\geq \Prob 
	\left[
		\forall h \in {\cal H}_i, \;
        \frac{ (\gamma_h - \gamma_{h-1})
            \left| 
                \phi_i^{[A][h]}-\phi_i^{[h]}
            \right|
        }{ (\gamma_h - \gamma_{h-1})
            \phi_i^{[h]}
        }
            < \varepsilon  
	\right]
= \Prob 
	\left[
		\forall h \in {\cal H}_i, \;
        \frac{
            \left| 
                \phi_i^{[A][h]}-\phi_i^{[h]}
            \right|
        }{
            \phi_i^{[h]}
        }
            < \varepsilon  
	\right] \\
\displaystyle
= 
    1-
    \Prob 
    \left[\exists h \in {\cal H}_i, 
		\frac{| \phi_i^{[A][h]}-\phi_i^{[h]}| }{ \phi_i^{[h]}}
            \geq \varepsilon
	\right] 
\geq 
    1-\sum_{h \in {\cal H}_i}
    \Prob 
    \left[ 
		\left| \phi_i^{[A][h]}-\phi_i^{[h]} \right|  
            \geq \varepsilon \cdot \phi_i^{[h]} 
	\right]  \\
\displaystyle
\geq 1- \sum_{h \in {\cal H}_i}
    \Prob 
	\left[ 
		\left| \phi_i^{[A][h]}-\Exp \left[ \phi_i^{[A][h]} \right] \right| 
            \geq \left( \frac{{\varepsilon}}{n(n-1)} \right)
	\right]  \\
\displaystyle 
\geq 1-  2\sum_{h \in {\cal H}_i}
    \exp 
    \left( 
        \frac{-2M 
        \left( 
            \frac{{\varepsilon}}{n(n-1)}
        \right)^2
        }{
            (n-1)^2
        } 
    \right)
=1- 2\sum_{h \in {\cal H}_i}
    \exp \left( -M \frac{2  \varepsilon^2}{n^2(n-1)^4} \right) \\
\geq  
1- 2H \exp \left( - \ln \left( \frac{2H}{\delta} \right) \right) =
  1-\delta.
 \end{array}
\] 

\medskip 
\noindent
(2) If we set 
	$\displaystyle M\geq \frac{n^2(n-1)^4 \ln (2nH/\delta)}{2\varepsilon^2}$,
	then we have that
\[
\begin{array}{l}
\displaystyle 
\Prob \left[
				\forall i \in N, 
					\frac{\left| 
						\phi^A_i	- \phi_i 
					\right|}{\phi_i}
					< \varepsilon 
			\right]
\geq \Prob 
	\left[
		\forall i \in N, 
        \forall h \in {\cal H}_i,
		\frac{ \left| \phi_i^{[A][h]}-\phi_i^{[h]} \right| 
        }{ \phi_i^{[h]}
        }
        < \varepsilon
    \right]  \\
\displaystyle 
= 1- \Prob 
	\left[
		\exists i \in N, 
        \exists h \in {\cal H}_i,
		\frac{ \left| \phi_i^{[A][h]}-\phi_i^{[h]} \right| 
        }{ \phi_i^{[h]}
        }
        \geq \varepsilon
    \right]  \\
\displaystyle 
\geq 1- \sum_{i\in N} \sum_{h \in {\cal H}_i}
    \Prob 
	\left[ 
		\left| \phi_i^{[A][h]}-\Exp \left[ \phi_i^{[A][h]} \right] \right| 
            \geq \left( \frac{{\varepsilon}}{n(n-1)} \right)
	\right]  \\
\displaystyle
\geq 1- \sum_{i \in N} \sum_{h \in {\cal H}_i}
2 \exp 
        \left( 
            \frac{-2M 
            \left( 
                \frac{{\varepsilon}}{n(n-1)} 
            \right)^2 
            }{
                {(n-1)}^2
            }
        \right) 
= 1- \sum_{i\in N} \sum_{h \in {\cal H}_i}
    2 \exp 
        \left( -M 
            \frac{2 \varepsilon^2 }{n^2(n-1)^4}
        \right) \\
\displaystyle
 \geq 
 1- 2nH  \exp 
	\left( 
		- \ln \left( \frac{2nH}{\delta} \right) 
 	\right) 
= 1-\delta.  
\end{array}
\]
\end{proof}

\section{Computational Experiments} \label{section:comp}

We conducted computational experiments 
    to evaluate the accuracy and efficiency
    of the proposed Monte Carlo-based approximation method 
    for computing the Shapley value in MCST games. 
In particular, 
    we empirically examined how the sample size required 
    for a prescribed accuracy depends 
    on the problem size and approximation parameters.

For instance generation,
    we used structural results due to Ando~\cite{ando2012computation}, 
    which show that 
    when the underlying graph is chordal 
    and edge weights are restricted to $\{0,1\}$, 
    the Shapley value of the corresponding MCST game 
    can be computed in polynomial time. 
This property enables us to compute 
    the exact Shapley value efficiently 
    and to use it as ground truth in our experiments. 
Specifically, we first generated a chordal graph 
    in which all edges have weight zero. 
Then, by adding edges of weight one, 
    we constructed a complete graph whose edge weights belong to $\{0,1\}$. 
By this procedure, we obtained 0-1 weighted complete graph instances 
    while retaining the ability to compute 
    the exact Shapley value via the chordal structure 
    prior to edge completion.

The number of players was set to $n = 3,4,5,\ldots,10$. 
For each value of $n$, we generated instances 
    on which Player~1 is not a null player. 
Instance generation was repeated until three such instances were obtained 
    for each $n$. 
For each generated instance, the exact Shapley value 
    of Player~1 in the corresponding saving game 
    was computed and taken as ground truth.

To approximate the Shapley value, 
    we applied the Monte Carlo algorithm 
    with an increasing number of samples. 
The sample size $M$ was initialized at $M=100$ 
    and increased in increments of $100$. 
For each value of $M$, the approximation procedure 
    was repeated independently $20$ times. 
In each trial, the relative error
    between the estimated Shapley value and the exact value 
    was computed.

For a given accuracy parameter $\varepsilon \in \{0.9,0.8,\ldots,0.1\}$,
    we counted the number of trials 
    in which the relative error did not exceed $\varepsilon$, 
    and defined the success rate 
    as the ratio of such trials 
    to the total number of repetitions. 
Fixing $\delta = 0.25$, 
    we regarded a sample size $M$ as sufficient 
    if the success rate was at least $1-\delta$ 
    for all generated instances simultaneously. 
The minimum sample size $M$ 
    satisfying this condition was recorded 
    for each value of $\varepsilon$.

Finally, we visualized the results by plotting
    the relationship between the sample size $M$ and $1/\varepsilon^{2}$, 
    as well as the relationship 
    between the number of players $n$ 
    and the empirically observed sample size 
    required to achieve the prescribed accuracy.

These results confirm that 
    the observed sampling complexity 
    is consistent with the theoretical bounds 
    and to demonstrate 
    the practical efficiency 
    of the proposed approximation scheme.

\subsection{Relationship between the Sample Size $M$ and $1/\varepsilon^{2}$}

In this subsection, we investigate the relationship between the sample size $M$ and the accuracy parameter $\varepsilon$, focusing on whether the empirical behavior of the Monte Carlo method is consistent with the theoretical bounds derived for the FPRAS.

According to the theoretical analysis, the number of samples required to approximate the Shapley value within a relative error $\varepsilon$ with high probability is proportional to $1/\varepsilon^{2}$. To verify this relationship empirically, we fixed the number of players $n$ and varied the accuracy parameter $\varepsilon$ over the range $\varepsilon = 0.9, 0.8, \ldots, 0.1$.

For each value of $\varepsilon$, we increased the sample size $M$ incrementally and evaluated the approximation performance using the procedure described in the previous subsection. For a given $M$, the approximation was repeated independently $20$ times for each instance, and the relative error was computed in each trial. The success rate was defined as the fraction of trials in which the relative error did not exceed $\varepsilon$.

Fixing $\delta = 0.25$, we regarded a sample size $M$ as sufficient if the success rate was at least $1-\delta$ for all generated instances simultaneously. The minimum value of $M$ satisfying this condition was recorded for each $\varepsilon$.

By plotting the obtained sample size $M$ against $1/\varepsilon^{2}$, we examine whether a linear relationship emerges, as predicted by the theoretical analysis. This comparison allows us to assess the validity of the asymptotic sample complexity bound in practical settings.

\begin{figure}[htbp]
  \centering
  \includegraphics[width=0.7\linewidth]{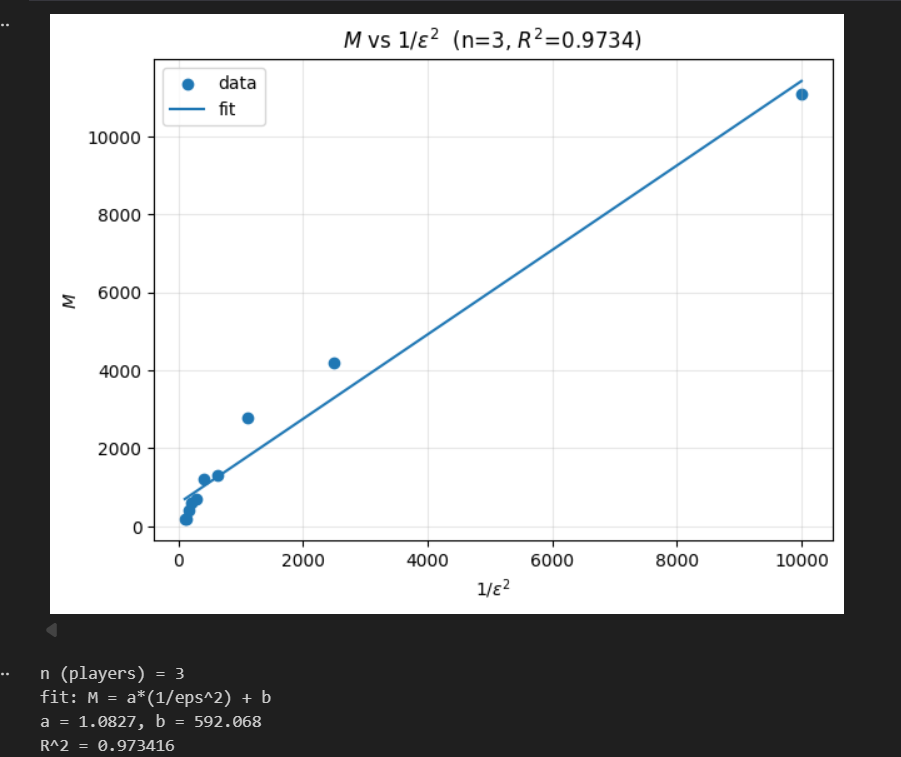}
  \caption{Relationship between the sample size $M$ and $1/\varepsilon^{2}$($n=3$).}
\end{figure}

\subsection{Relationship between the Number of Players $n$ and the Required Sample Size $M$}

In this subsection, we examine how the sample size required to achieve a prescribed accuracy depends on the number of players $n$, and compare the empirical results with the theoretical sample complexity bound.

For a fixed accuracy parameter $\varepsilon$, we varied the number of players $n$ and evaluated the minimum sample size $M$ that satisfies the success criterion described in the previous subsection. According to the theoretical analysis of the Monte Carlo approximation scheme, the sample size required to guarantee a relative error of at most $\varepsilon$ with probability at least $1-\delta$ is bounded by
$\mbox{\(\left\lceil \frac{n^2 (n-1)^4 \ln(2/\delta)}{2\varepsilon^2} \right\rceil\)}$.

To assess the validity of this bound in practice, 
    we compared the theoretical sample size with 
    the result obtained from computational experiments. 
For each value of $\varepsilon$, 
    we plotted the number of players $n$ on the horizontal axis, 
    and on the vertical axis we plotted 
    the natural logarithm of the theoretical sample size 
    given by the above formula, 
    as well as the natural logarithm 
    of the empirically observed sample size $M$.

By visualizing both quantities on a logarithmic scale, we aim to evaluate whether the growth rate of the experimentally required sample size with respect to $n$ is consistent with the theoretical prediction. This comparison provides insight into the tightness of the theoretical bound and the scalability of the proposed approximation method as the problem size increases.

\begin{figure}[htbp]
  \centering
  \includegraphics[width=0.7\linewidth]{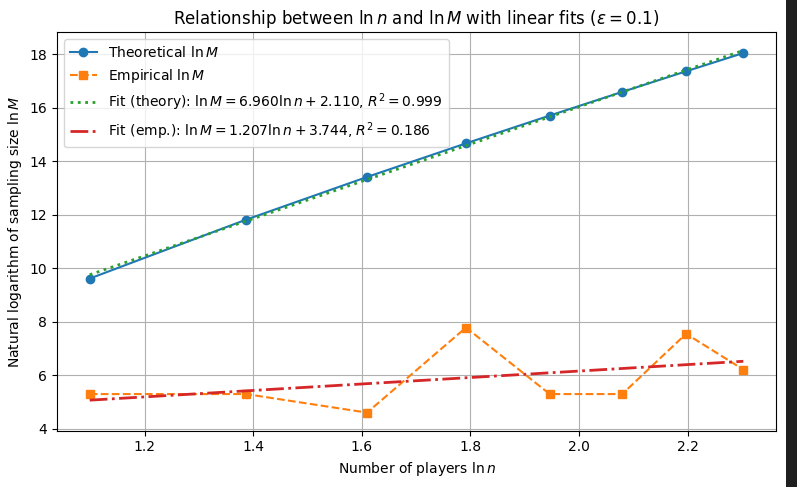}
  \caption{Relationship between the sample size $M$ and $n$($\varepsilon=0.1$).}
  \label{n=3}
\end{figure}

\section{Summary of Contributions and Conclusions}
\label{section:conclusion}

\subsection{MCST Game} \label{subsection:MCSTgames}

In this subsection, we summarize some results corresponding 
    MCST game obtained in this paper.

\begin{Cor}
Let $(N, c)$ be an MCST game defined by 
    non-negatively weighted complete graph 
    $\langle \ComG [N_r], \Vec{w} \rangle$.
Denote $\varphi$ be the Shapley value of $(N, c)$.
Every player $i\in N$ satisfies $\varphi_i \leq w(r,i)$.
For any player $i \in N$, 
    the following statements are equivalent:
\begin{description}
\item[\rm (0)]  the Shapley value of player $i \in N$ 
    satisfies $\varphi_i=w(r,i),$
\item[\rm (1)] player $i$ is a dummy player in $(N, c)$
    {\rm (i.e., $\forall S \subseteq N\setminus \{i\}, c(S \cup \{i\})=c(S)+c(\{i\})$)}, 
\item[\rm (2)] $\forall j \in N\setminus \{i\}$, $c(\{i,j\})=c(\{i\})+c(\{j\})$, 
\item[\rm (3)] $\forall j \in N\setminus \{i\}$,  
    $w(r,i)\leq w(i,j) \geq w(r,j)$. 
\end{description} 

\noindent
When we consider a simple MCST game $(N, c)$ defined by
    0-1 weighted complete graph
    $\langle \ComG [N_r], \widetilde{\Vec{w}} \rangle$,
    we have the following.
\begin{description}
\item[\rm (4)] A player $i \in N$ is not a dummy player in $(N, c)$, 
    if and only if   $i$ satisfies \\
 $[\exists j \in N\setminus \{i\},$  
 $(\widetilde{w}(r, i),\widetilde{w}(i,j),\widetilde{w}(r, j))
    \in \{(1,0,0), (1,0,1), (0,0,1)\}]$.
\item[\rm (5)] If a player  $i \in N$ is not a dummy player in $(N, c)$, 
    then $\varphi_i \leq w(r,i)- \frac{1}{n(n-1)}$. 
\end{description} 
\end{Cor}

We describe Algorithm~\ref{algo-Cost}
    for computing Approximate Shapley value $\varphi^A$
    for the MSCT game $(N,c)$. 
    
\begin{algorithm}
\caption{}\label{algo-Cost}
\begin{algorithmic}[Algorithm 1]
\Input  
     non-negatively weighted complete graph $\langle \ComG [N_r], \Vec{w}\rangle$
     and a positive integer $M$.
\Output Approximate Shapley value $\varphi^A_i$ 
    for the MSCT game $(N,c)$ 
    defined by the non-negatively weighted complete graph  $\langle \ComG [N_r], \Vec{w}\rangle$. 
    \medskip 
    
\State Execute Algorithm~\ref{algo1} 
    with the number of samples set to $M$,
      and obtain a vector~$\phi^A$.
\State Output $\varphi^A_i =w(r,i)-\phi^A_i$ for each $i \in N$.
\end{algorithmic}
\end{algorithm}

Then we have the following results.

\begin{Cor}
Let $\varphi^A=(\varphi^A_1, \varphi^A_2,\ldots , \varphi^A_n)$
    denote the vector obtained by Algorithm~\ref{algo-Cost}.
Suppose that Algorithm~\ref{algo-Cost} is applied 
    to a simple MCST game $(N,c)$
    defined by  $\langle \ComG [N_r], \widetilde{\Vec{w}} \rangle$.
Then, for any $\varepsilon >0$ and \mbox{$0< \delta <1$}, 
    the following statements hold.

\begin{description}
\item[{\rm (1)}]
If we set 
	$\displaystyle M \geq \frac{n^2 (n-1)^4 \ln (2/\delta)}{2\varepsilon^2}$,
	then each player $i \in N$ satisfies that
\[
\Prob 
	\left[ 
		\frac{| \varphi^A_i-\varphi_i| }{\widetilde{w}(r,i)-\varphi_i}
            < \varepsilon  
	\right] \geq 1-\delta.
\] 

\item[{\rm (2)}]
 If we set 	
	$\displaystyle M\geq \frac{n^2 (n-1)^4 \ln (2n/\delta)}{2\varepsilon^2}$,
	then \
\[
 \Prob 
	\left[
		\forall i \in N, 
             \frac{| \varphi^A_i-\varphi_i| }{\widetilde{w}(r,i)-\varphi_i}
            < \varepsilon  
	\right] \geq 1-\delta.
\]
\end{description}

\noindent
Next, consider an MCST game $(N, c)$ defined by
    a non-negatively weighted complete graph
    $\langle \ComG [N_r], \Vec{w} \rangle$.
Let $H$ denote the number of mutually distinct 
    positive edge-weights, that is, 
    $H$ is the size of the set 
   $\{\gamma \in \mathbb{R} 
    \mid \exists \{i,j\}\subseteq N_r, \gamma=w(i,j)>0 \}$.
Then, for any $\varepsilon >0$ and \mbox{$0< \delta <1$,}
    the following statements hold.
    
\begin{description}
\item[{\rm (3)}]
If we set 
	$\displaystyle M \geq \frac{n^2(n-1)^4 \ln (2H/\delta)}{2\varepsilon^2}$,
	then each player $i \in N$ satisfies that
\[
\Prob 
	\left[ 
		\frac{| \varphi^A_i-\varphi_i| }{w(r,i)-\varphi_i}
            < \varepsilon  
	\right] \geq 1-\delta.
\] 

\item[{\rm (4)}]
 If we set 	
	$\displaystyle M\geq \frac{n^2(n-1)^4 \ln (2nH/\delta)}{2\varepsilon^2}$,
	then 
\[
 \Prob 
	\left[
		\forall i \in N, 
             \frac{| \varphi^A_i-\varphi_i| }{w(r,i)-\varphi_i}
            < \varepsilon  
	\right] \geq 1-\delta.
\]
\end{description}
\end{Cor}

\subsection{Concluding Remarks}

In this paper, we studied the problem 
    of approximating the Shapley value 
    in minimum cost spanning tree (MCST) games 
    by introducing and analyzing 
    the associated MCST-saving games. 
We first established a necessary and sufficient condition 
    for identifying null players 
    and derived lower bounds 
    on the Shapley value of non-null players. 
These structural results enabled 
    an efficient reduction of the problem 
    by eliminating null players 
    in a pre-processing step.

Building on these properties, 
    we developed a Monte Carlo 
    based fully polynomial-time 
    randomized approximation scheme (FPRAS) 
    for computing the Shapley value 
    in MCST-saving games. 
The proposed method achieves 
    a multiplicative (relative-error) 
    approximation guarantee, 
    in contrast to existing 
    Monte Carlo approaches 
    that provide only additive-error bounds. 
We further analyzed the computational complexity 
    of the algorithm. 

Future research directions include 
    extending the proposed approach to broader classes 
    of cooperative games arising from network design problems
    and improving sampling efficiency 
    through the use of saving games

\bibliography{refer}

\end{document}